\documentclass[oneside,english]{amsart}
\usepackage[T1]{fontenc}
\usepackage[utf8]{inputenc}
\usepackage{color}
\definecolor{note_fontcolor}{rgb}{0.800781, 0.800781, 0.800781}
\usepackage{prettyref}
\usepackage{mathrsfs}
\usepackage{mathtools}
\usepackage{amstext}
\usepackage{amsthm}
\usepackage{amssymb}
\usepackage{cite}

\usepackage{changes}

\numberwithin{equation}{section}
\numberwithin{figure}{section}

\theoremstyle{definition}
 \newtheorem*{defn*}{\protect\definitionname}
 
\theoremstyle{plain}
 \newtheorem{thm}{\protect\theoremname}
 \newtheorem{prop}{\protect\propositionname}
 \newtheorem{lem}{\protect\lemmaname}

\newcommand{\pref}{\prettyref}

\newcommand{\Z}{\mathbf{Z}}
\renewcommand{\H}{\mathfrak{H}}
\newcommand{\R}{\mathcal{R}}
\renewcommand{\L}{\mathcal{L}}
\newcommand{\M}{\mathcal{M}}
\newcommand{\F}{\mathbf{F}}
\newcommand{\Rb}{\mathbf{R}}

\newcommand{\mc}[1]{\begin{pmatrix*}[c]#1\end{pmatrix*}}

\usepackage{babel}
  \providecommand{\definitionname}{Definition}
  \providecommand{\lemmaname}{Lemma}
  \providecommand{\propositionname}{Proposition}
\providecommand{\theoremname}{Theorem}

\begin{document}
\global\long\def\pd#1#2{\frac{\partial#1}{\partial#2}}
\global\long\def\weight#1{\operatorname{weight}(#1)}

\title[]{Formal recursion operators of integrable nonevolutionary equations and Lagrangian systems}

\author[ACQ]{Agustín Caparrós Quintero}
\author[RHH]{Rafael Hernández Heredero}
\address{ETSI Sistemas de Telecomunicaci\'on\\
Universidad Politécnica de Madrid}
\email[ACQ]{agustin.caparros.quintero@alumnos.upm.es}
\email[RHH]{rafahh@etsist.upm.es}

\begin{abstract} We derive the general structure of the space of formal recursion operators of nonevolutionary equations~$q_{tt}=f(q,q_{x},q_t,q_{xx},q_{xt},q_{xxx},q_{xxxx})$. This allows us to classify integrable Lagrangian systems with a higher order Lagrangian of the form~$\mathscr{L}=\frac12 L_2(q_{xx}, q_x, q)\,q_t^2 + L_1(q_{xx}, q_x, q)\, q_{t} + L_0(q_{xx}, q_x, q)$.  The key technique relays on exploiting a homogeneity of the determining equations of formal recursion operators. This technique allows us to extend the main results to more general equations~$q_{tt}=f(q,q_{x},\ldots,q_{n};q_{t},q_{xt},\ldots,q_{mt})$.
\end{abstract}

\maketitle

\section{Introduction and motivation}

In this paper we research integrability properties of equations of the form
\begin{equation}\label{eq:eq4}
q_{tt}=f(q,q_{x},q_t,q_{xx},q_{xt},q_{xxx},q_{xxxx}),\qquad\pd{f}{q_{xxxx}}\neq0
\end{equation}
with particular interest in those arising as Euler-Lagrange equations of Lagrangian densities
\begin{equation}\label{eq:lagrangian}
\mathscr{L}=\frac12 L_2(q_{xx}, q_x, q)\,q_t^2 + L_1(q_{xx}, q_x, q)\, q_{t} + L_0(q_{xx}, q_x, q).
\end{equation}
This family of Lagrangian systems includes interesting examples of integrable models with physical relevance like
\begin{align}
\mathscr{L}&=\displaystyle\frac12\left(q_t+q_{xx}-q_x^2\right)^2, &&\text{NLS}\label{eq:eqNLS}
\\[3mm]
\mathscr{L}&=\displaystyle\frac12\left(q_t+q_{xx}\right)^2+\frac12q_x^3,&&\text{Boussinesq}\label{eq:eqBouss}
\\[3mm]
\mathscr{L}&=\displaystyle\frac{\left(q_t+q_{xx}-\frac12R'(q)\right)^2}{4\left(q_x^2-R(q)\right)}+\frac1{12}R''(q),\ R^{\rm v}=0.
&&{\text{Landau-Lifshitz}}\label{eq:eqLL}
\end{align}
which are related respectively to the nonlinear Schr\"odinger equation, the (potential) Boussinesq equation and a continuous Landau-Lifshitz model for spin chains~\cite[p.~1454]{AMS},\cite{ASY,QC1}.

A very powerful theory to study the integrability of differential equations is the~\emph{symmetry approach to integrability}~\cite{MSS,MSY}, that started dealing with evolutionary equations $u_t=f(x,t;q,q_x,\ldots,q_n)$. Under its umbrella some studies have obtained  results on more general nonevolutionary equations including~\pref{eq:eq4}. In~\cite{NovWang,MNW}, very general results were obtained using a perturbative symmetry approach in symbolic representation. The price to pay for such generality is that this representation, in principle, requires the equation to be polynomial and homogeneous. Another extension of the symmetry approach applicable to general non-evolutionary equations, not necessarily polynomial, was presented in letter~\cite{HSS}. The main aim of that letter was to study third order systems of the form~$u_{tt}=u_{xxx}+F(q,q_x,q_t,q_{xx},q_{xt})$.

In the present work we expand the ideas in~\cite{HSS} to treat equations of second order in~$t$ and any order in~$x$, including the fourth order ones~\pref{eq:eq4}. The fourth order case shows a more complicated symmetry structure and requires a further elaboration of the theory. We have developed a simple theoretical framework based on the homogeneity of the determining equations for a recursion operator. This framework allows us to characterise the structure of the space of recursion operators and the space of symmetries of a PDE
\begin{equation}\label{eq:eqi}
q_{tt}=f(q,q_{x},\ldots,q_{n};q_{t},q_{xt},\ldots,q_{mt})
\end{equation}
of in fact, any differential order in~$x$. For the fourth order case~\pref{eq:eq4}, we give explicit integrability conditions on the rhs~$f$ and we use them to give a classification of integrable Lagrangian systems~\pref{eq:lagrangian}.

In Section~2 we give an introduction to the theory. The main object related to integrability in our approach is a~\emph{formal recursion operator}. We review the derivation given in~\cite{HSS} of the determining equations of formal recursion operators for general non-evolutionary equations~\pref{eq:eqi}.

In Section~3 we display the coarse, general structure of the space of formal recursion operators for a given equation~\pref{eq:eqi}, discussing their main properties regarded as pseudodifferential operators. We show how this space is a field and discuss the equivalence between two different representations, either using scalar recursion operators or using matrix recursion operators. We give the definition of integrable differential equation under this approach.

In Section~4 we describe the finer structure of the space of formal recursion operators in the case of fourth order equations~\pref{eq:eq4}. Exploiting the homogeneity of the determining equations and using a symbolic notation allows us to proof that in this case this space, regarded as an algebra, is generated by two independent elements.  Furthermore, we explain how explicit integrability conditions appear in the form of conservation laws  for equations of type~\prettyref{eq:eq4}.

In Section~5 we give, as an application of the previous theory, a classification  of integrable Lagrangian systems~\eqref{eq:lagrangian}. For each retrieved case, we find in~\pref{sec:recops} non-formal recursion operators and the structure of the space of higher symmetries, thus proving their integrability.

In the conclusions Section~7 we discuss how the homogeneity scheme developed to study equations~\pref{eq:eq4} can be extended to general equations~\pref{eq:eq}.

\section{Integrability of equations~$q_{tt}=f(q,q_{t},q_{x},q_{xt},\ldots)$}
\label{sec:intint}

Lie symmetries and, especially, higher or generalised Lie symmetries~\cite{Olver} have been prolifically used to characterise integrability of partial differential equations. The symmetry approach to integrability is a well established theory that provides a rigorous, computable characterisation of integrability for  evolution equations
\[
q_t=f(x,t;q,q_x,\ldots,q_n)
\]
on two independent variables~$x$, $t$ and one dependent variable~$q=q(x,t)$, with the notations
\[
q_t\coloneqq\frac{\partial q}{\partial t},\quad q_x\coloneqq\frac{\partial q}{\partial x},\quad q_{i}\coloneqq\frac{\partial^{i}q}{\partial x^{i}},\quad i\in\mathbf{N}\quad\text{and}\quad q_{0}\coloneqq q.
\]
A generalised symmetry is an evolution equation in an additional parameter~$\tau$
\[q_\tau=g(x,t;q,q_x,\ldots,q_m)
\]
that is compatible with the given equation. We will refer to such a symmetry using the differential function~$g$, called~\emph{characteristic} of the symmetry.  According to the symmetry approach, an equation is integrable if it admits an infinite number of symmetries of arbitrarily high differential order. Usually, a sufficient condition for the integrability of a PDE is to admit a~\emph{recursion operator}~$\R$. A recursion operator generates a whole hierarchy of symmetries starting from a seed symmetry~$g_0$, i.e.~all iterations~$\R^i(g_0)$ are symmetries.
In order to study integrability, working with recursion operators is much easier than working directly with symmetries.

The equations studied in this paper are differential equations  of the form~\pref{eq:eqi}
\begin{equation}\label{eq:eq}
q_{tt}=f(q,q_{x},\ldots,q_{n};q_{t},q_{xt},\ldots,q_{mt})
\end{equation}
where
\[
 q_{tt}\coloneqq\frac{\partial^{2}q}{\partial t^2},\quad q_{jt}\coloneqq\frac{\partial^{j+1}q}{\partial x^{j}\partial t},\quad j\in\mathbf{N},\quad\text{and}\quad q_{0t}\coloneqq q_{t}.
\]

We consider, within the framework of differential algebra, a set~$\H$ of differential functions~$h(q,q_{x},\ldots,q_{k};q_{t},q_{xt},\ldots,q_{lt})$ with two derivations~$D\coloneqq D_{x}$ and~$D_{t}$. If~$\partial h/\partial q_{k}\neq0$ and~$\partial h/\partial q_{lt}\neq0$ the pair~$(k,l)$ will denote the differential order of~$h$. 
The derivations~$D$ and~$D_{t}$ act over differential functions~$h\in\H$ as total derivatives
\begin{gather*}
Dh=\sum_{i=0}^{k}\frac{\partial h}{\partial q_{i}}q_{i+1}+\sum_{j=0}^{l}\frac{\partial h}{\partial q_{jt}}q_{j+1t},\\
D_{t}h=\sum_{i=0}^{k}\frac{\partial h}{\partial q_{i}}q_{i+1t}+\sum_{j=0}^{l}\frac{\partial h}{\partial q_{jt}}q_{j+1tt}\doteq\sum_{i=0}^{k}\frac{\partial h}{\partial q_{i}}q_{i+1t}+\sum_{j=0}^{s}\frac{\partial h}{\partial q_{jt}}D^{j+1}f
\end{gather*}
where the symbol~$\doteq$ denotes an equality that holds only in the solution space of~\prettyref{eq:eq}, i.e.~$q_{tt}$ and its $x$-derivatives~$q_{jtt}$ are to be substituted by~$f$ and its corresponding total derivatives.

Generalised Lie symmetries of~\prettyref{eq:eq} are of the form
\begin{equation}\label{eq:sym}
q_{\tau}=g(q,q_{x},\ldots,q_{r};q_{t},q_{xt},\ldots,q_{st})
\end{equation}
with~$g\in\H$. A generalised symmetry~$g$ provides an additional derivation~$D_{\tau}$ in the differential algebra that commutes with~$D$ and acts on a differential function~$h$ as
\[
D_{\tau}h=\sum_{i=0}^{k}\frac{\partial h}{\partial q_{i}}D^{i}g+\sum_{j=0}^{l}\frac{\partial h}{\partial q_{jt}}D^{j}D_{t}g.
\]
We will use indistinctly the notations~$D_{\tau}h=h_{\tau}=D_{\tau}(h)$.
The symmetry condition, i.e.~the compatibility of~\pref{eq:eq}--\pref{eq:sym} is then
\begin{equation}\label{eq:symc}
D_{\tau}f\doteq D_{t}^{2}g
\end{equation}
and it can be expressed in terms of the linearisation, or variational equation,
of~\prettyref{eq:eq}. We write this variational equation as
\begin{equation}\label{eq:ldeq}
D_{t}^{2}\phi=\left(U+VD_{t}\right)\phi
\end{equation}
where~$\phi=\phi(x,t)$ is a new dependent variable, the variation, and
\[
\mathcal{F}\coloneqq U+VD_{t}
\]
is the linearisation operator of the differential function~$f$, being
$U$ and~$V$ differential operators on~$D$ only:
\begin{equation}\label{eq:symbUV}
U\coloneqq u_{n}D^{n}+u_{n-1}D^{n-1}+\cdots+u_{0},\quad V\coloneqq v_{m}D^{m}+v_{m-1}D^{m-1}+\cdots+v_{0}
\end{equation}
with
\[
u_{i}\coloneqq \pd{^{i}f}{q_{i}},\quad v_{j}\coloneqq \pd f{q_{jt}}.
\]
The symmetry condition~\prettyref{eq:symc} can now be written as
\begin{equation}\label{eq:symcond1}
D_{t}^{2}g\doteq\left(U+VD_{t}\right)g
\end{equation}
providing a variational interpretation to generalised symmetries~\pref{eq:sym}.
 
We are interested in equations that are~\emph{integrable}
from the symmetry approach point of view, i.e.~equations that admit an infinite number of symmetries~\prettyref{eq:sym} of arbitrarily high differential order~$(r,s)$. To treat this problem, it is convenient to use an alternative way to write the symmetry condition. It consists in performing a second linearisation on~\prettyref{eq:symcond1} or, equivalently, considering the compatibility of the linearised versions
of~\prettyref{eq:eq} and~\prettyref{eq:sym}. The linearised version of~\prettyref{eq:eq} is~\prettyref{eq:ldeq}, and that of~\prettyref{eq:sym}
is
\begin{equation}
\phi_{\tau}=\left(L+MD_{t}\right)\phi,\label{eq:ldsym}
\end{equation}
where~$\mathcal{G}\coloneqq L+MD_{t}$ is the linearisation operator of the differential function~$g$, i.e.~
\[
L=l_{r}D^{r}+l_{r-1}D^{r-1}+\cdots+l_{0},\quad M=m_{s}D^{s}+m_{s-1}D^{s-1}+\cdots+m_{0}
\]
are differential operators with
\[
l_{i}\coloneqq\pd g{q_{i}},\quad m_{j}\coloneqq\pd g{q_{jt}}.
\]
Now the compatibility of equations~\prettyref{eq:ldeq} and~\prettyref{eq:ldsym} leads to symmetry conditions that can be written as equations over differential operators:
\begin{equation}\label{eq:symcond}
\begin{gathered}L_{tt}-VL_{t}+[L,U]+2M_{t}U+MU_{t}+[M,V]U\doteq U_{\tau},\\
M_{tt}-VM_{t}+2M_{t}V+[M,U]+2L_{t}+[L,V]+MV_{t}+[M,V]V\doteq V_{\tau}.
\end{gathered}
\end{equation}
In this formula we have defined the extension of a derivation~$D_{\tau}$ (or~$D_{t}$,~$D_{x}$) to a differential operator~$O=o_{k}D^{k}+o_{k-1}D^{k-1}+\cdots+o_{0}$ as
\[
O_{\tau}=D_{\tau}(o_{k})D^{k}+D_{\tau}(o_{k-1})D^{k-1}+\cdots+D_{\tau}(o_{0}).
\]
Equations~\prettyref{eq:symcond} have to be interpreted as the system of PDE's obtained equating every coefficient of~$D^{j}$ in the lhs's to the corresponding coefficient in the rhs's. Equations~\prettyref{eq:symcond} can be used to find symmetries of a fixed differential order~$(r,s)$. In this case, the operators~$U$ and~$V$ are determined by the equation, so they are data in~\prettyref{eq:symcond}, while the operators~$L$ and~$M$ are the unknowns. This procedure is completely equivalent to directly finding symmetries from~\pref{eq:symcond1}. In the problem of finding integrable equations all four operators~$U$,~$V$ (associated to the equation~\prettyref{eq:eq}) and~$L$,~$M$ (associated to the symmetry~\prettyref{eq:sym}) are unknowns in~\prettyref{eq:symcond}. But the fact that we require the existence of symmetries of an arbitrarily high order, i.e.~that the operators~$L$ and~$M$ have to be of an arbitrarily high degree in~$D$, allows to simplify the problem. For an arbitrary pair of differential operators~$L$ and~$M$ of degrees~$r$ and~$s$, the left hand sides of~\prettyref{eq:symcond} have generically degrees higher than~$r$ and~$s$. But the right hand sides are
\[
U_{\tau}=D_{\tau}(u_{n})D^{n}+\cdots+D_{\tau}(u_0),\quad V_{\tau}=D_{\tau}(v_{m})D^{m}+\cdots+D_{\tau}(v_0)
\]
and have generically a fixed degree of~$n$ and~$m$ respectively. If there are symmetries of higher order~$r\gg n$ and~$s\gg m$, the number of conditions to be met in the left hand sides increases when increasing~$r$ and~$s$, tending to be infinite in the arbitrarily high order case. The symmetry approach then proposes~\cite{HSS} to consider the operator equations
\begin{gather}
L_{tt}-VL_{t}+[L,U]+2M_{t}U+MU_{t}+[M,V]U\doteq0,\label{eq:fro1}\\
M_{tt}-VM_{t}+[M,U]+2M_{t}V+[L,V]+2L_{t}+MV_{t}+[M,V]V\doteq0.\label{eq:fro2}
\end{gather}
In order to make this idea to work the operators~$L$ and~$M$ that solve~\pref{eq:fro1}--\pref{eq:fro2} must be taken in the class of~\emph{pseudodifferential} operators instead of differential operators. Equations~\pref{eq:eq} for which there exist nontrivial solutions to~\pref{eq:fro1}--\pref{eq:fro2} are deemed as integrable under the symmetry approach, because it is a necessary condition for having symmetries of arbitrarily high order.

\section{Formal recursion operators}

The algebraic structure of the space of solutions~$\mathcal{R}\coloneqq L+MD_{t}$
of~\prettyref{eq:fro1}--\prettyref{eq:fro2} is much richer than that of~\prettyref{eq:symcond}.
The product of two solutions,
defined as
\begin{multline*}
\left(L_{1}+M_{1}D_{t}\right)\cdot\left(L_{2}+M_{2}D_{t}\right)=
 L_{1}\cdot L_{2}+M_{1}\cdot\left(L_{2}\right)_{t}+M_{1}\cdot M_{2}\cdot U
\\
 {}+\left[L_{1}\cdot M_{2}+M_{1}\cdot L_{2}+M_{1}\cdot\left(M_{2}\right)_{t}+M_{1}\cdot M_{2}\cdot V\right]D_{t}
\end{multline*}
is also a solution, a reflection of the fact that the set of generalised symmetries is a Lie algebra. The identity operator~$\mathcal{I}=1+0D_{t}$ is always a solution. Once we have a product and an identity, we can have~\emph{inverses}. The inverse of a differential operator exists in the field of~\emph{pseudodifferential} operators. We will consider operators
\begin{equation}\label{eq:scfro}
\R=L+MD_{t}
\end{equation}
where~$L$ and~$M$ are  pseudodifferential operators of degrees $r$ and~$s$
\begin{equation}\label{eq:psLM}
\begin{gathered}
L=l_{r}D^{r}+l_{r-1}D^{r-1}+\cdots+l_{0}+l_{-1}D^{-1}+\cdots,\\
M=m_{s}D^{s}+m_{s-1}D^{s-1}+\cdots+m_{0}+m_{-1}D^{-1}+\cdots
\end{gathered}
\end{equation}
being~$l_{i}$, $m_{j}$ differential functions. The product of pseudodifferential operators follows the standard integration by parts rule
\begin{equation}\label{eq:dm1}
D^{-1}\cdot\mathcal{O}=\mathcal{O}D^{-1}-D^{-1}\cdot D(\mathcal{O})D^{-1}=\sum_{i=0}^{\infty}(-1)^iD^i(\mathcal{O})D^{-1-i}.
\end{equation}
In the class of pseudodifferential operators every nonzero operator~$\mathcal{R}$ has a unique inverse~$\mathcal{R}^{-1}$ and, if~$\mathcal{R}$ is a solution of~\prettyref{eq:fro1}--\prettyref{eq:fro2}, then~$\mathcal{R}^{-1}$ is also a solution. These solutions~$\mathcal{R}$ are called \emph{formal recursion operators} (briefly fro's) of the given differential equation~\prettyref{eq:eq}. This name is justified because they satisfy the equation
\begin{equation}\label{eq:gro}
\mathcal{H}\left(L+MD_{t}\right)\doteq\left(\overline{L}+MD_{t}\right)\mathcal{H}
\end{equation}
where~$\mathcal{H}\coloneqq D_{t}^{2}-U-VD_{t}$ is the linearisation operator of~\prettyref{eq:eq} and~$\overline{L}\coloneqq L+2M_{t}+[M,V]$. This is the equation that a recursion operator satisfies~\cite{Olver}. The formal character of~$\mathcal{R}$ is due to the fact that it is generally written as a pseudodifferential infinite series (as it can be done with~\pref{eq:dm1}). An expression for~$\mathcal{R}$ in closed, finite form is what is usually called a (proper) recursion operator.

The proofs of the assertions in the previous paragraph are easily done using an alternative, \emph{matrix} formulation of the problem. It is based on representing the differential equation~\pref{eq:eq} as a first order system
\begin{equation}\label{eq:eqsy}
\begin{aligned}
q_t&=p\\
p_t&=f(q,q_x,\ldots,q_r,p,p_x,\ldots,p_s)
\end{aligned}
\end{equation}
and performing the double linearisation process described in the previous section. The variational equations of~\pref{eq:eqsy} can be written in a matrix form as
\begin{gather*}
\mc{\phi\\\psi}_t\!=\F\mc{\phi\\\psi}\text{\quad with\quad}\F\coloneqq\mc{0&1\\U&V}
\end{gather*}
and the variational equations of~\pref{eq:sym} as
\begin{gather}
\label{eq:Rm}\mc{\phi\\\psi}_{\tau}\!=\Rb\mc{\phi\\\psi}\text{\quad with\quad}\Rb\coloneqq\mc{L&M\\\widehat{L}&\widehat{M}}
\\
\notag
\text{and\qquad $\widehat{L}\coloneqq L_t+MU$,\qquad $\widehat{M}\coloneqq M_t+MV+L$.}
\end{gather}
The commutation of these two equations gives rise to the usual expression for a recursion operator
\begin{equation}\label{eq:mfro}
\mathbf{R}_{t}=[\mathbf{F},\mathbf{R}]
\end{equation}
but in matrix form. 

Using~\pref{eq:mfro} it is straightforward to prove the following basic features of the set~$\mathfrak{F}$ of formal matrix recursion operators of~\prettyref{eq:eq}, that is, the space of solutions~$\Rb$ of~\prettyref{eq:mfro} given a fixed equation~\prettyref{eq:eq}, i.e.~a fixed~$\mathbf{F}$.
\begin{prop} Let~$\mathfrak{F}$ the set of formal matrix recursion operators of equation~\pref{eq:eq}, i.e.~the set of solutions~$\Rb$ to equation~\pref{eq:mfro}. Then
\begin{enumerate}
\item The set~$\mathfrak{F}$ is a linear space over the constants.
\item The product of two fro's~$\mathbf{R}_{1},\mathbf{R}_{2}\in\mathfrak{F}$
is a fro~$\mathbf{R}_{1}\cdot\mathbf{R}_{2}\in\mathfrak{F}$.
\item For any fro~$\mathbf{R}\in\mathfrak{F}$,~$\mathbf{R}\neq0$ there exists an inverse fro~$\mathbf{R}^{-1}\in\mathfrak{F}$.
\end{enumerate}
Thus~$\mathfrak{F}$ is a field and a graded algebra. 
\end{prop}
The addition, product and inverses of matrix fro's correspond univocally to the same operations in scalar fro's, and equation~\pref{eq:mfro} is completely equivalent to equations~\pref{eq:fro1}--\pref{eq:fro2}. Thus there is an isomorphism between scalar and matrix fro's. Every formal (scalar) recursion operator~$\R=L+MD_t$ has a unique associated formal~\emph{matrix} recursion operator~$\Rb$ that is a solution of~\pref{eq:mfro} and has the form given in~\pref{eq:Rm}. We will use from now on only scalar recursion operators~$\R$, because they present some computational advantages.

From the discussion in~\pref{sec:intint} it is clear that there is a relation between higher order generalised symmetries and formal recursion operators. If a generalised symmetry~$g$ has a sufficiently high order, its linearisation~$\mathcal{G}=L+MD_{t}$
anihilates many terms in the equations of a fro~\prettyref{eq:fro1}--\prettyref{eq:fro2}, and it gives the leading behaviour of a fro of the same order. Thus, although a generalised symmetry does not generate a full formal recursion operator because the rhs's of~\prettyref{eq:symcond} are not zero, it partially solves equations~\prettyref{eq:fro1}--\prettyref{eq:fro2}.
This fact is formalised by the following definition, where~$\deg(\mathcal{O})$ denotes the degree of a pseudodifferential operator~$\mathcal{O}$, i.e.~its maximum power in~$D$.
\begin{defn*}
The pseudodifferential operator~$\mathcal{R}=L+MD_{t}$ with~$\deg(L)=r$, $\deg(M)=s$ is a \emph{formal symmetry }(or approximate symmetry) of rank~$(k,l)$, of equation~\prettyref{eq:eq}, if the left hand sides of~\prettyref{eq:fro1}--\prettyref{eq:fro2}
are of degree  at most~$r+n-k$ and~$s+m-l$ respectively.
\end{defn*}
The next result follows.
\begin{prop}
The linearisation~$\mathcal{G}=L+MD_{t}$ of a symmetry~\prettyref{eq:sym}
of equation~\prettyref{eq:eq} is a formal symmetry of rank~$(r,s)$.
\end{prop}

The space of formal symmetries of a given rank has the same properties than the space of fro's. It is a linear space over the constants, the product of two formal symmetries of the same rank is a formal symmetry of the same rank, and the inverse of a formal symmetry is a formal symmetry of the same rank. 
\begin{prop}
If an equation~\pref{eq:eq} admits infinite higher symmetries of arbitrarily high order, then it admits a non-trivial formal recursion operator~$\R=L+MD_t$, i.e.~there exist nontrivial local solutions~\pref{eq:psLM} to~\pref{eq:fro1}--\pref{eq:fro2}.
\end{prop}

Formulas~\pref{eq:fro1}--\pref{eq:fro2} can be used to solve the following two problems.
\begin{itemize}
\item Determining if a given equation~\prettyref{eq:eq} is integrable. It could be achieved by finding an admitted formal recursion operator~$\mathcal{R}= L+MD_{t}$, solving~\prettyref{eq:fro1}--\prettyref{eq:fro2}. The operators~$L$ and~$M$
are the unknowns in the equations, whereas~$U$ and~$V$ are data
given by the rhs~$F$ of the given equation.
\item Classifying a family of equations of the form~\prettyref{eq:eq}. The operators~$U$ and~$V$ become unknowns too. In some cases (for some given~$n$, and~$m$) this problem is solvable. In the next section we will study a particular case.
\end{itemize}

\section{Integrable equations of the type~$q_{tt}=f(q,q_{x},q_{xx},q_{xxx},q_{xxxx};q_{t},q_{xt})$}

The finer structure of the space of fro's strongly depends on the
differential order of the rhs of~\prettyref{eq:eq}, i.e.~on~$n$ and~$m$. The case of some third order equations ($n=3$,~$m=1$) was already discused in \cite{HSS}. In the present paper we consider the case of equations with~$n=4$,~$m=1$ i.e.~of the form~\pref{eq:eq4}
\[
q_{tt}=f(q,q_{x},q_{xx},q_{xxx},q_{xxxx};q_{t},q_{xt}),\qquad\pd{f}{q_{xxxx}}\neq0
\]
that produces a more complex structure, and is the one relevant to the classification of Lagrangian systems~\pref{eq:lagrangian}. 

We write the generic formal recursion operator as a pseudodifferential operator~\pref{eq:scfro}. For computational reasons and without loss of generality we initially set~$\deg(L)=r= k$ and~$\deg{M}=s=k-1$, with~$k\in\Z$, writing
\begin{gather*}L=l_{k}D^{k}+l_{k-1}D^{k-1}+\cdots+l_{i}D^{i}+\cdots,\\
M=m_{k-1}D^{k-1}+m_{k-2}D^{k-2}+\cdots+m_{j}D^{j}+\cdots.
\end{gather*}
The first integrability conditions from~\prettyref{eq:fro1} and~\prettyref{eq:fro2}
are the following.

The highest nonzero coefficient in~\prettyref{eq:fro2} is that of
$D^{k+2}$, and the resulting equation is:
\begin{equation}
(k-1)m_{k-1}D_{x}u_{4}-4u_{4}D_{x}m_{k-1}=0.\label{eq:cond21}
\end{equation}
Multiplying by an integrating factor~$u_{4}^{-(k-1)/4-1}$ allows
to write 
\begin{equation}
D_{x}\left[u_{4}^{-(k-1)/4}m_{k-1}\right]=0.\label{eq:obsm1}
\end{equation}
This implies that 
\begin{equation}
m_{k-1}=\mu_{k-1}u_{4}^{(k-1)/4}\label{eq:mkm1}
\end{equation}
 where~$\mu_{k-1}$ is a constant.

The next nonzero coefficient in~\prettyref{eq:fro2} is that of~$D^{k+1}$,
implying that
\begin{multline}
(k-2)m_{k-2}D_{x}u_{4}-4u_{4}D_{x}m_{k-2}-3u_{3}D_{x}m_{k-1}-6u_{4}D_{xx}m_{k-1}\\
+\frac{1}{2}(k-1)m_{k-1}\left[(k-1)D_{xx}u_{4}+2D_{x}u_{3}-D_{xx}u_{4}\right]=0\label{eq:cond22}
\end{multline}
Substituting the value found for~$m_{k-1}$ and using the integrating
factor~$u_{4}^{-(k-2)/4-1}$ yields
\[
D_{x}\left[u_{4}^{-\left(k-2\right)/4}m_{k-2}\right]=\mu_{k-1}\left[\frac{k-1}{4}D_{x}\left(\frac{u_{3}}{u_{4}^{3/4}}\right)+\frac{(k-1)(k-5)}{2}D_{xx}\left(u_{4}^{1/4}\right)\right].
\]
Thus
\begin{equation}
m_{k-2}=\mu_{k-2}u_{4}^{\left(k-2\right)/4}+\mu_{k-1}u_{4}^{\left(k-2\right)/4}\left[\frac{k-1}{4}\frac{u_{3}}{u_{4}^{3/4}}+\frac{(k-1)(k-5)}{8}\frac{D_{x}u_{4}}{u_{4}^{3/4}}\right],\label{eq:mkm2}
\end{equation}
being~$\mu_{k-2}$ a constant. This constant appears because of the linearity of the equations, and introduces a redundancy in our solution~$M$ ($k$ is arbitrary). For illustrative reasons we keep it (and analogous constants) in the next few paragraphs.

The highest nonzero coefficient in the other equation~\prettyref{eq:fro1}
corresponds to~$D^{k+3}$, yielding
\begin{multline}
kl_{k}D_{x}u_{4}-4u_{4}D_{x}l_{k}+2u_{4}D_{t}m_{k-1}+m_{k-1}D_{t}u_{4}\\
-u_{4}v_{1}D_{x}m_{k-1}+(k-1)u_{4}m_{k-1}D_{x}v_{1}=0.\label{eq:cond11}
\end{multline}
Using again an integrating factor~$u_{4}^{-(k-1)/4-1}$ and the expression for~$m_{k-1}$, it is found that
\[
D_{x}\left[u_{4}^{-k/4}l_{k}\right]=\mu_{k-1}\left[\frac{k-1}{4}D_{x}\left(\frac{v_{1}}{u_{4}^{1/4}}\right)-\frac{k+1}{2}D_{t}\left(\frac{1}{u_{4}^{1/4}}\right)\right]
\]
A first integrability condition appears here:~$\rho_{0}\coloneqq u_{4}^{-1/4}$
must be the density of a conservation law
\begin{equation}
D_{t}\left(\frac{1}{u_{4}^{1/4}}\right)=D_{x}\sigma_{0}.\label{eq:cl0}
\end{equation}
This is an obstruction that narrows the class of equations~\prettyref{eq:eq}
that are integrable. Supposing that this condition is met, we set 
\begin{equation}
l_{k}=\lambda_{k}u_{4}^{k/4}+\mu_{k-1}\left[\frac{k-1}{4}u_{4}^{(k-1)/4}v_{1}-\frac{k+1}{2}u_{4}^{k/4}\sigma_{0}\right].\label{eq:lk1}
\end{equation}

The coefficient in~\prettyref{eq:fro1} of~$D^{k+2}$, after substituting the definition of a second integrability condition
\[
D_{t}\left(\frac{u_{3}}{u_{4}}\right)=D_{x}\sigma_{1},
\]
implies that
\begin{multline*}
D_{x}\left[u_{4}^{-(k-1)/4}l_{k-1}\right]=
\mu_{k-2}\left[\frac{k-2}{4}D_{x}\left(\frac{v_{1}}{u_{4}^{1/4}}\right)-\frac{k}{2}D_{x}\sigma_{0}\right]\\
+\lambda_{k}\left[\frac{1}{4}kD_{x}\left(\frac{u_{3}}{u_{4}^{3/4}}\right)+\frac{k(k-4)}{2}D_{x}^{2}\left(u_{4}^{1/4}\right)\right]
\\
+\mu_{k-1}\left[
\frac{\left(k^{2}-3k-2\right)\left(k-1\right)}{4}D_{x}\left(u_{4}^{1/4}D_{x}\sigma_{0}\right)\right.
\\
{}+\frac{\left(k-1\right)^{2}}{16}D_{x}\left(\frac{u_{3}v_{1}}{u_{4}}\right)+\frac{\left(k^{2}-5k+5\right)\left(k-1\right)}{32}D_{x}\left(\frac{v_{1}D_{x}u_{4}}{u_{4}}\right)
\\
-\frac{k(k+1)}{8}D_{x}\left(\frac{\sigma_{0}u_{3}}{u_{4}^{3/4}}\right)+\frac{(4-k)k(k+1)}{4}D_{x}^{2}\left(u_{4}^{1/4}\sigma_{0}\right)+\frac{k+1}{8}D_{x}\sigma_{1}
\\
\left.\phantom{\frac{\left(k^{2}-3k+4\right)k}{32}}+\frac{k-1}{4}D_{x}v_{0}+\frac{\left(k-1\right)(k-5)}{8}D_{xx}v_{1}
\right]
\end{multline*}
and we can find the expression of~$l_{k-1}$ accordingly, in the form
$l_{k-1}=\lambda_{k-1}u_4^{(k-1)/4}+\lambda_{k}\left[\cdots\right]+\mu_{k-2}\left[\cdots\right]+\mu_{k-1}\left[\cdots\right]$.

The coefficient in~\prettyref{eq:fro2} of~$D^{k}$ leads to another
conservation law
\[
D_{t}\left(\frac{v_{1}}{\sqrt{u_{4}}}-2\frac{\sigma_{0}}{u_{4}^{1/4}}\right)=D_{x}\sigma_{2}
\]
and to an expression of~$m_{k-3}$ of the form~$m_{k-3}=\mu_{k-3}u_{4}^{(k-3)/4}+\lambda_{k}\left[\cdots\right]+\mu_{k-2}\left[\cdots\right]+\mu_{k-1}\left[\cdots\right]$.

The computation continues producing an infinite system of equations that can be solved with a recursive scheme. In fact, in~\prettyref{eq:fro1}--\prettyref{eq:fro2} the terms with highest order in~$D$ come from~$[L,U]$ and~$[M,U]$ and produce equations where the unknowns~$m_{i}$,~$l_{j}$ appear in one term  that can be isolated in the lhs of two equations of the form
\begin{equation}\label{eq:eqsml}
\begin{aligned}D_{x}\left(u_{4}^{-i/4}l_{i}\right)&=A_{i},\\
D_{x}\left(u_{4}^{-j/4}m_{j}\right)&=B_{j},\end{aligned}\qquad\begin{aligned}
 i&=k, k-1,\ldots,\\ j&=k-1, k-2,\ldots.
\end{aligned}
\end{equation}
The right hand sides~$A_{i},$~$B_{j}$ of these equations must be
total derivatives, so there appear an infinite number of integrability obstructions (that we will see are in the form of conservation laws) and arbitrary constants of integration in each of the terms of~$L$ and~$M$. 
We will justify a posteriori that these obstructions do not depend on~$k$.

In practice, one can implement the procedure taking~$k=0$ and solving~\prettyref{eq:eqsml} in the order~$j=0,$~$j=-1$,~$i=1$,~$i=0$ and then alternating~$j=-l$,~$i=-l+1$ for~$l\geq2$. Of course, to completely solve the problem of integrability in this way one should perform an infinite number of steps. Usually, a small number of steps quickly reduces the possible subfamilies of equations~\prettyref{eq:eq} admitting fro's, and other procedures are used to prove full integrability of these remaining cases. For example, finding an explicit recursion operator (and seed or root symmetries) is enough to prove integrability.

In the following results we deal with the space of formal recursion operators of an equation~\pref{eq:eq4} that satisfies all the integrability conditions imposed by~\pref{eq:eqsml}, and~$\lambda_i$ and~$\mu_j$ are the integration constants arising from its integration.

\begin{prop} A basis of the vector space the formal recursion operators of an integrable equation~\pref{eq:eq4} is formed by the operators
\[
\begin{gathered}\mathcal{L}_{i}\coloneqq\left\{ L+MD_{t}:\ \lambda_{i}=1;\ \lambda_{i'}=0,\,i'\neq i;\ \mu_{j}=0\right\}, \\
\mathcal{M}_{j}\coloneqq\left\{ L+MD_{t}:\ \lambda_{i}=0;\ \mu_{j}=1;\ \mu_{j'}=0,\,j'\neq j\right\} ,
\end{gathered}
\qquad 
i,j\in\Z.
\]
\end{prop}
The leading coefficients of~$\mathcal{L}_{i}$ and~$\M_{j}$ are
\begin{align}
\label{eq:Li}
\mathcal{L}_{i}= & u_{4}^{i/4}D^{i}+\left[\frac{i}{4}u_{3}u_{4}^{\frac{i-1}{4}-\frac{3}{4}}+\frac{\left(i-4\right)i}{2}u_{4}^{\frac{i-1}{4}}D_{x}(u_{4}^{1/4})\right]\,D^{i-1}+\cdots 
\\
\nonumber
& +\left\{ \left(\frac{i}{4}u_{4}^{\frac{i-4}{4}}v_{1}-\frac{i}{2}\sigma_{0}u_{4}^{\frac{i-3}{4}}\right)\,D^{i-3}+\cdots\right\} D_{t},
\\
\label{eq:Mj}
\mathcal{M}_{j}= & \left[\frac{j}{4}v_{1}u_{4}^{\frac{j}{4}}-\frac{j+2}{2}\sigma_{0}u_{4}^{\frac{j+1}{4}}\right]\,D^{j+1}+\cdots \\
\nonumber
& +\left\{ u_{4}^{\frac{j}{4}}\,D^{j}+\left[\frac{j}{4}u_{3}u_{4}^{\frac{j-4}{4}}+\frac{\left(j-4\right))j}{2}D_{x}(u_{4}^{1/4})u_{4}^{\frac{j-1}{4}}\right]\,D^{j-1}+\cdots\right\} D_{t}
\end{align}

We have the following important technical results.
\begin{lem}
\label{lem:one}
\ \\[-4mm]
\begin{enumerate}
\item $\mathcal{L}_{i}=\mathcal{L}_{1}^{i}$ for all~$i\in\mathbf{Z}$.
\item $\mathcal{M}_{-2}^{2}=\mathcal{L}_{0}=\mathcal{I}$.
\item\label{enu:comm} $\M_{-2}\L_1=\L_1\M_{-2}=\M_{-1}$.
\item $\M_j=\M_{-2}\L_{j+2}$,\quad $\L_i=\M_{-2}\M_{i-2}$.
for all~$i,j\in\mathbf{Z}$.
\end{enumerate}
\end{lem}

\begin{proof}
A sketch of the proof is as follows. Equations~\prettyref{eq:fro1}--\prettyref{eq:fro2}
have an underlying homogeneity, and their solutions can be expressed as sums of homogeneous solutions. The weights of the homogeneity can be defined as follows
\begin{equation}\label{eq:weights}
\begin{gathered}
\weight D=w_{x},\quad\weight{D_t}=w_{t},\\
\weight{u_{i}}=2w_{t}-iw_{x},\quad\weight{v_{j}}=w_{t}-jw_{x}.
\end{gathered}
\end{equation}
Any equation  provided by the coefficient of some power~$D^{i}$ in~\prettyref{eq:fro1} or~\prettyref{eq:fro2} is homogeneous, and thus the operators~$\L_i$ and~$\M_j$ are homogeneous. Expressions~\pref{eq:Li}--\pref{eq:Mj} imply that~$\weight{\mathcal{L}_{i}}=(2w_t-4w_x)\frac{i}{4}+iw_x$ and~$\weight{\mathcal{M}_{j}}=(2w_t-4w_x)\frac{j}{4}+jw_x+w_t$ i.e.
\begin{equation}\label{eq:wlm}
\weight{\mathcal{L}_{i}}=\frac{i}{2}w_{t},\quad\weight{\mathcal{M}_{j}}=\frac{j+2}{2}w_{t}.
\end{equation}
Thus,~$\weight{\mathcal{L}_{i}}=\weight{\mathcal{M}_{i-2}}$. 
Let us prove (\pref{enu:comm}). Homegeneity implies that
\[
\L_{1}\M_{-2}=\alpha\mathcal{L}_{1}+\beta\mathcal{M}_{-1},\qquad\alpha,\beta\text{ constants}.
\]
Using~\pref{eq:Li}--\pref{eq:Mj} and multiplying
\[
\L_{1}\M_{-2}=-\left(\frac{v_1}{4 u_4^{1/4}}+\frac{\sigma_{0}}{2}\right)D^0+\cdots+\left[\frac{1}{u_4^{1/4}}\, D^{-1} +\cdots\right]D_{t} 
\]
so~\prettyref{eq:Li}--\prettyref{eq:Mj} imply that~$\alpha=0$ and~$\beta=1$. And the leading terms of~$\M_{-2}\L_{1}$ coincide with those of~$\mathcal{L}_{1}\mathcal{M}_{-1}$, so~$\mathcal{M}_{-1}\mathcal{L}_{1}=\M_0$ too. The remaining proofs are analogous. 
\end{proof}
For ease of notation, in what follows we continue using the notation~$s^o+s^{o-1}+\cdots$ to denote a formal series~$\sum_{i=0}^{\infty} s^{o-i}$.
\begin{lem}
\label{lem:AnyPse} Any pseudodifferential operator~$\mathcal{O}$
of degree~$(r,s)$
\[
\mathcal{O}=o_{r}D^{r}+o_{r-1}D^{r-1}+\cdots+\left[p_{s}D^{s}+p_{s-1}D^{s-1}+\cdots\right]D_{t}
\]
can be written as
\[
\mathcal{O}=a_{k}\mathcal{L}_{1}^{k}+a_{k-1}\mathcal{L}_{1}^{k-1}+\cdots+\left[b_{l}\mathcal{M}_{0}\mathcal{L}_{1}^{l}+b_{m-1}\mathcal{M}_{0}\mathcal{L}_{1}^{l-1}+\cdots\right]D_{t}
\]
with~$a_{i}$,~$b_{j}$ appropriate diferential functions.
\end{lem}
\begin{proof}
If~$r>s+1$ the operator~$\mathcal{O}-o_{r}u_{4}^{-r/4}\L_{1}^{r}$
has degree~$(r-1,s)$. If~$r\leq s+1$ the operator~$\mathcal{O}-p_{s}u_{4}^{-s/4}\M_{-2}\L_1^{s+2}=\mathcal{O}-p_{s}u_{4}^{-s/4}\M_{0}\L_1^{s}$ is of order~$(r,s-1)$. Proceeding by induction implies that the sum of all the subtracted operators is~$\mathcal{O}$. 
\end{proof}
\begin{thm}\label{thm:one}
Any formal recursion operator~$\R=L+MD_{t}$ of an integrable equation~\pref{eq:eq4} can be expressed as
\[
\R=\alpha_{k}\mathcal{L}_{1}^{k}+\alpha_{k-1}\mathcal{L}_{1}^{k-1}+\cdots+\left[\beta_{l}\mathcal{M}_{0}\mathcal{L}_{1}^{l}+\beta_{l-1}\mathcal{M}_{0}\mathcal{L}_{1}^{l-1}+\cdots\right]D_{t}
\]
with~$\alpha_{i}$,~$\beta_{j}$ appropriate constants.
\end{thm}

\begin{proof}
If~$r>s+1$ equation~\prettyref{eq:Li} implies that the coefficient~$l_{r}$
must be of the form~$l_{r}=\alpha_{r}u_{4}^{r/4}$, with~$\alpha_{r}$
constant. Linearity implies that~$\R-\alpha_{r}\L_{1}^{r}$ is a
recursion operator. The case~$r\leq s+1$ is treated analogously,
and induction implies that all the coefficients~$a_{i}$,~$b_{j}$
of~\pref{lem:AnyPse} are constant.
\end{proof}
\begin{thm}\label{thm:two}
Given a differential equation~\pref{eq:eq4} the conditions of integrability, i.e.~the conditions of existence of local formal recursion operators~$\L_{1}$,~$\M_{0}$, have the form of conservation laws
\[
D_{t}\rho_{i}=D_{x}\sigma_{i}.
\]
\end{thm}
\begin{proof}
\pref{lem:AnyPse} implies that, in particular, both~$D_{x}$
and~$D_{t}$ can be written as
\begin{gather*}
D_{x}=\rho_{1}\mathcal{L}_{1}+\rho_{0}+\rho_{-1}\mathcal{L}_{1}^{-1}+\cdots+\pi_{-2}\mathcal{M}_{0}\mathcal{L}_{1}^{-2}+\pi_{-3}\mathcal{M}_{0}\mathcal{L}_{1}^{-3}+\cdots,\\
D_{t}=\sigma_{1}\mathcal{L}_{1}+\sigma_{0}+\sigma_{-1}\mathcal{L}_{1}^{-1}+\cdots+\theta_{0}\mathcal{M}_{0}+\theta_{-1}\mathcal{M}_{0}\mathcal{L}_{1}^{-1}+\cdots
\end{gather*}
Commutativity of~$D_{x}$ and~$D_{t}$ implies the result.
\end{proof}
We give in Appendix~\pref{app:tableint}, for further reference, a table of the first integrability conditions for equations in the class~\pref{eq:eq4}.

The homogeneity scheme can be used to study the symmetry structure of any integrable equation of the type~\pref{eq:eq}. There exist two possible structures: i) an algebra of fro's generated by two independent operators, when~$m$ is even; ii) an algebra of fro's generated by one operator, when~$m$ is odd. Cf.~the conclusions section~\ref{sec:concl} for more details.

\section{Classification of integrable Lagrangian systems}

As an application of the previous theory, we will perform a classification that has interest on its own: finding integrable Lagrangian systems with Lagrangians
\begin{equation}
\mathscr{L}=\frac12 L_2(q_{xx}, q_x, q)\,q_t^2 + L_1(q_{xx}, q_x, q)\, q_{t} + L_0(q_{xx}, q_x, q).
\end{equation}
The corresponding Euler-Lagrange equations 
\[\frac{\delta \mathscr{L}}{\delta q}=-D_t\left(\pd{\mathscr{L}}{q_{t}}\right)+D_x^2\left(\pd{\mathscr{L}}{q_{xx}}\right)-D_x\left(\pd{\mathscr{L}}{q_{x}}\right)+\pd{\mathscr{L}}{q}
\]
are of the form~\pref{eq:eq4}, linear in~$q_{xxxx}$. A non-degeneracy condition is
\begin{equation}\label{eq:cu4}
\pd{f}{q_{xxxx}}=
\frac{1}{L_2}\left(\frac12\pd{^2L_2}{q_{xx}^2}q_t^2+
    \pd{^2L_1}{q_{xx}^2}q_t+\pd{^2L_0}{q_{xx}^2}\right)\neq0
\end{equation}
so~$L_2\neq0$ and $L_{2}$, $L_1$ and~$L_{0}$ cannot be
simultaneously linear in~$q_{xx}$:
\begin{equation}
\label{eq:coeffcond}
\left|\pd{^2L_2}{q_{xx}^2}\right|+\left|\pd{^2L_1}{q_{xx}^2}\right|+\left|\pd{^2L_0}{q_{xx}^2}\right|\neq0\quad\text{(identically).}
\end{equation}
The classification process is too long to be described in full detail here, so we just explain the general procedure and remark especial, crucial points that arose in the computation.

We have to solve the infinite set of partial differential equations derived from the integrability conditions
\begin{equation}\label{eq:intc}
c_i=D_t\rho_i-D_x\sigma_i=0,\qquad i=0,\ldots.
\end{equation}
The first condition is~$c_0=D_t\rho_0-D_x\sigma_0=0$ where, as listed in Appendix~\pref{app:tableint},~$\rho_0$ is the first canonical density
\[
\rho_0
=\left(\frac{1}{L_2}\pd{^2\mathscr{L}}{q_{xx}^2}\right)^{-1/4}
=L_2^{1/4}\left(\frac12\pd{^2L_2}{q_{xx}^2}q_t^2+\pd{^2L_1}{q_{xx}^2}q_t+\pd{^2L_0}{q_{xx}^2}\right)^{-1/4}.
\]
The highest order derivatives in~$D_t\rho_0$ are~$q_{xxxx}$ and~$q_{xxt}$, implying that~$\sigma_0$ is a function with dependencies~$\sigma_0(q_{xxx},q_{xt},q_{xx},q_t,q_x,u)$. The coefficients of~$c_0$ in~$q_{xxxx}$ and~$q_{xxt}$ simply determine the dependency of~$\sigma_0$ in higher derivatives. The first genuine obstruction is obtained from the coefficient of~$c_0$ in~$q_{xxx}^2q_t^3$:
\[
\left(L_2\pd{^3L_2}{q_{xx}^3}+3\pd{^2L_2}{q_{xx}^2}\pd{L_2}{q_{xx}}
\right)\pd{^2L_2}{q_{xx}^2}=0
\]
whose general solution is
\[
L_{2}(q_{xx},q_x,q)=\sqrt{L_{22}(q_x,q)q_{xx}^2 + 
    L_{21}(q_x,q)q_{xx} +L_{20}(q_x,q)}.
\]
Exploring further condition~$c_0$ forces the discriminant~$
L_{21}^2-4L_{22}L_{20}$ to be zero, so
\[
L_{2}(q_{xx},q_x,q)= L_{21}(q_x,q)q_{xx} +L_{20}(q_x,q).
\]
An important step to simplify the computations is to realise that~$L_{21}(q_x,q)$ can be taken as zero. Supposing~$L_{21}(q_x,q)\neq0$, and using conditions~$c_0$ and~$c_1$ one can prove that the term of the differential equation in~$q_{xxxx}$ is of the form
\[
\frac{1}{L_2}\pd{^2\mathscr{L}}{q_{xx}^2}q_{xxxx}=\frac{q_{xxxx}}{[Dh(q_x,q)+k]^4},\qquad\text{$k$~constant.}
\]
A contact transformation with~$\overline{x}=\xi=h(q_x,q)+k\,x$,~$\overline{q}=\psi(q_x,q)$ and
\[
\left(\pd{h}{q}q_x+k\right)\,\pd{\psi}{q_x}=\pd{h}{q_x}q_x\,\pd{\psi}{q}
\]
simplifies this term to~$f(q_x,q)q_{xxxx}$, that is to say, we can take~$L_{21}(q_x,q)=0$. Notice that we have used a transformation that depends explicitly on~$x$, but that, if the original equation does not depend explicitly on~$x$, neither will the transformed equation. This first transformation that we use to classify the integrable classes in~\pref{eq:lagrangian} exhausts the possibility of further using a proper contact transformation as equivalence relation. But we will be free to use~\emph{point} transformations, that we call~\emph{allowed transformations} to choose representatives in the integrable classes.

We just deduced that we can restrict our study to Lagrangians
\[\mathscr{L}=\frac12L_{20}(q_x,q)q_t^2+L_1(q_{xx},q_x,q)q_t+L_0(q_{xx},q_x,q).
\]
The condition~$c_0$ still gives restrictions, e.g.~implying that~$L_1$ is linear:
\[L_1(q_{xx},q_x,q)=L_{11}(q_x,q)q_{xx}+L_{10}(q_x,q).
\]
Without loss of generality~$L_{11}(q_x,q)$ can be taken as~$0$ by equivalence of Lagrangians. Combining~$c_0$ and~$c_1$ one finds that~$L_0$ is quadratic in~$q_{xx}$:
$$L_0(q_{xx},q_x,q)=L_{02}(q_x,q)q_{xx}^2+L_{01}(q_x,q)q_{xx}+L_{00}(q_x,q).$$
and $L_{01}(q_x,q)$ can be taken too as~$0$ by equivalence. The Lagrangian at this stage is
\[\mathscr{L}=\frac12L_{20}(q_x,q)q_t^2+ L_{10}(q_x,q)q_t+L_{02}(q_x,q)q_{xx}^2+L_{00}(q_x,q).
\]
with~$L_{02}(q_x,q)\neq0$ because of~\pref{eq:coeffcond}. The condition~$\partial c_0/{\partial q_{xx}}=0$ implies the following differential relation between~$L_{02}(q_x,q)$  and~$L_{20}(q_x,q)$:
\[
\left(\pd{}{q_x}+\frac14 L_{20}^{-1}\pd{L_{20}}{q_x}\right)^2z=0
\]
with~$z=L_{02}^{-1/4}$.
Its general solution is
\begin{equation}\label{eqL22}
L_{02}(q_x,q)=\frac{L_{20}(q_x,q)}{[\lambda(q)q_x+\mu(q)]^4}
\end{equation}
where~$\lambda(q)$ and~$\mu(q)$ are arbitrary functions not simultaneously zero. Now~$c_0$ is solved setting~$\mu(q)=\mu$ constant, while~$c_1$ implies that
\begin{equation}\label{eq::eqL20}
\pd{^3}{q_x^3}\left(\frac{\lambda(q)q_x+\mu}{L_{20}}\right)=0.
\end{equation}
In general
\begin{equation}\label{eq:L20a}
L_{20}(q_x,q)=\frac{\lambda(q)q_x+\mu}{L_{202}(q)q_x^2+ 
L_{201}(q)q_x+ L_{200}(q)}.
\end{equation}
To simplify computations, we considered two main cases according to the value~$\mu=0$ or~$\mu\neq0$. The classification proceeds by branching inside these two main cases depending on the value of the resultant between the polynomials on~$q_x$ in the denominator and numerator~of~\pref{eq:L20a}. It is important to simplify the intermediate families of equations by using allowed transformations.  We produce the final classification ordering the cases according to the form of~$L_{20}(q_x,q)$, which transforms  under point transformations~$\overline{t}=\chi(t)$, $\overline{x}=\xi(x,q)$, $\overline{q}=\psi(x,q)$ as
\begin{equation}\label{eq::pointtr}
\overline{L}_{20}=\frac{\chi'\left(\xi_{q}q_x+\xi_{x}\right)}{\left(\xi_{x}\psi_{q}-\xi_{q}\psi_{x}\right)^{2}}L_{20}.
\end{equation}
We give below a list of Lagrangians that satisfy the first ten integrability conditions~$c_i$, $i=0,\ldots 10$. Each of these Lagrangians is a representative of a class of equations equivalent under contact transformations and the addition of total derivatives. In the next section, we will study the complete integrability of the given classes. 
\begin{gather}
\label{eq:lagrD1f}\tag{L1}
\mathscr{L}=\frac{q_t^2}{2}
   +\epsilon\, q_xq_t
   +\frac{q_{xx}^2}{2} + \delta_2\frac{q_x^2}{2} + \delta_1\frac{q^2}{2}+\delta_0 q,
\\
\label{eq:lagrD2f}\tag{L2}
\mathscr{L}=\frac{q_t^2}{2}
   -q_x^2q_t
   +\frac{q_{xx}^2}{2}+\frac{q_x^4}{2},
\\
\label{eq:lagrD8f}\tag{L3}
\mathscr{L}=\frac{q_t^2}{2}
+\frac{q_{xx}^2}2
+\frac{q_x^3}2,
\\
\label{eq:lagrD3f}\tag{L4}
\mathscr{L}=\frac{q_t^2}{2}+a(q)\,q_xq_t+\frac{q_{xx}^2}{2q_x^4}+a'(q)q_x\log{q_x}+\frac{a^2(q)}{2}q_x^2+d(q),\\
\label{eq:lagrD5f}\tag{L5}
\mathscr{L}=\frac{q_t^2}{2}+\left(\frac{\gamma}{q_x}+\epsilon\, q_x\right)q_t+\frac{q_{xx}^2}{2q_x^4}+\frac{\epsilon^2}{2}q_x^2+\frac{\gamma^2}{2q_x^2}+\frac{\delta}{q_x},\quad|\gamma|+|\delta|\neq0,
\\
\label{eq:lagrD6f}\tag{L6}
\mathscr{L}=\frac{q_x}2q_t^2+(\epsilon q_x+\beta)q_xq_t+\frac{q_{xx}^2}{2q_x^3}+\frac{\epsilon^2}2q_x^3+\epsilon\beta\, q_x^2+\frac{\delta}{q_x},
\\
\label{eq:lagrD7f}\tag{L7}
\mathscr{L}=\frac{q_t^2}{2q_x}+\frac{b(q)}{q_x}\,q_t+\frac{q_{xx}^2}{2a(q)^4q_x^5}
+\frac{d_2(q)}{q_x},
\\
\label{eq:lagrD10f}\tag{L8}
\mathscr{L}=\frac{q_t^2}{2(q_x^2-1)} + \frac{q_{xx}^2}{2(q_x^2-1)} + d(q)q_x^2- \frac{d(q)}{3},\quad 
d'''(q)-8d(q)d'(q)=0.
\end{gather}
Notice that~\pref{eq:lagrD1f} is the linear system,~\pref{eq:lagrD2f} is equivalent to~\pref{eq:eqNLS},~\pref{eq:lagrD8f} is equivalent to~\pref{eq:eqBouss} and~\pref{eq:lagrD10f} is an alternative way of writing~\pref{eq:eqLL}.

\section{Recursion operators of integrable Lagrangian systems}
\label{sec:recops}

In this Section all the classes given in~\pref{eq:lagrD1f}--\pref{eq:lagrD10f} will be proven to be integrable except for classes~\pref{eq:lagrD3f} and~\pref{eq:lagrD7f}, that presented computational problems that are beyond the capability of our  computing facilities. For these two classes, nevertheless, we give examples that probably exhaust their integrable cases. From now on, when we say that a Lagrangian or a system admits a symmetry or recursion operator, we are of course referring to the corresponding Euler-Lagrange equations.

The proof of integrability is achieved by giving an explicit, non-formal recursion operator and corresponding seed symmetries. The recursion operator can be found by inspection in a few cases, but  usually this is quite difficult to do. The general approach that worked here is to use Sokolov's ansatz~\cite{Sok,DS} i.e.~to search for a recursion operator of the form
\[\R=\mathcal{D}+\sum_ks_kD^{-1}\cdot\mathcal{C}_k
\]
being~$\mathcal{D}$ a differential operator and the sum a linear combination where the coefficients~$s_k$ are symmetries and~$\mathcal{C}_k$ are linearisations of conserved densities~$\rho_k$.

\subsection{Lagr.~\pref{eq:lagrD1f}}
This is the linear system, and it admits as recursion operator any linear combination with constant coefficients of the operators~$D_x$ and $D_t$. The seed symmetry is~$q$, and every derivative of the field is a symmetry.

\subsection{Lagr.~\pref{eq:lagrD2f}}
This system is related to the NLS equation~\pref{eq:eqNLS} and simple inspection shows that~$\mathcal{M}_{-1}$ is a recursion operator:
\[\R=\mathcal{M}_{-1}=-q_x+2D^{-1}\cdot q_{xx} + D^{-1}D_t.
\]
This operator applied to the symmetry~$1$ produces~$q_x$, and applied again to the latter produces~$q_t$. Iterated,~$\R$  generates a basis of the whole symmetry hierarchy. The first higher symmetries in this case are
\[
\begin{array}{l}
    s^{\rm a}_1\coloneqq q_{xxx}+3q_xq_t-2q_x^3,\\
    s^{\rm b}_1\coloneqq q_{xxt}+2q_xq_{xxx}-\frac12q_{xx}^2+\frac32q_t^2+3q_x^2q_t-\frac72q_x^4.
\end{array}
\]
System~\pref{eq:lagrD2f} admits a symmetry algebra with a definite structure. We will refer to this symmetry structure as a \emph{standard} symmetry structure for equations~\pref{eq:eq}. The standard structure consists of an infinite sequence of symmetries of two types. One type is composed of symmetries with highest order term~$c_iq_i$, where~$i$ is an odd positive integer, and the other type with higher order term~$d_jq_{jt}$, where~$j$ is an even positive integer. We will use the followig notation to denote symmetries with different types of higher order terms:
\[
\begin{array}{lll}
\text{symmetries of type a:}& s^{\textrm{a}}_{i}=c_iq_{2i+1}+\cdots,&\ i=0,1,\ldots\\
\text{symmetries of type b:}& s^{\textrm{b}}_{j}=d_jq_{2jt}+\cdots,&\ j=0,1,\ldots
\end{array}
\]
where the ellipsis indicates lower order terms. The recursion operator acts over symmetries as~$\R(s^a_i)=s^b_i$ and~$\R(s^b_j)=s^a_{j+1}$, $i,j=0,1,\ldots$. Below we will show some systems with symmetry algebras being a variation to this general structure.

\subsection{Lagr.~\pref{eq:lagrD8f}}
This system is related to the potential Boussinesq equation~\pref{eq:eqBouss} and it admits the recursion operator~$\mathcal{M}_{1}$:
\[\R=\mathcal{M}_{1}=-\frac98q_t+\frac38D^{-1}\cdot q_{xt}+\left[D-\frac38q_xD^{-1}-\frac38D^{-1}\cdot q_x\right]D_t.
\]
Applied to~$-4/3$ it yields~$q_t$, applied to~$q_x$ it gives~$s^{\rm b}_1=q_{xxt}-3q_xq_t/2$, applied to this symmetry yields a symmetry~$s^{\rm a}_3=q_7+\cdots$, applied to~$q_t$, a symmetry~$s^{\rm a}_2=q_{5}+\cdots$, applied to the latter gives~$s^{\rm b}_3=q_{6t}+\cdots$ and so on. In general,~$\R(s^{\textrm{a}}_{i})=s^{\textrm{b}}_{i+1}$ and~$\R(s^{\textrm{b}}_{j})=s^{\textrm{a}}_{j+2}$, $i,j=0,1,\ldots$. Iterated,~$\R$  generates a basis of the whole symmetry hierarchy that differs from the standard structure explained for Lagr.~\pref{eq:lagrD2f} because it lacks all symmetries of the form~$s^{\textrm{a}}_{1+3i}$ and~$s^{\textrm{b}}_{2+3j}$ with~$i,j=0,1,\ldots$.
This ``defect'' is probably due to the fact that the potential Boussinesq system~\pref{eq:lagrD8f} is related to the usually called Boussinesq equation through a potentiation, a nonlocal transformation. The Boussinesq equation has more symmetries than the potential Boussinesq equation, and  some symmetries of the former could become nonlocal for the latter. But we are dealing here only with equations admitting local symmetries~\pref{eq:sym} and local recursion operators~\pref{eq:psLM}. More about this in the conclusions of Section~\ref{sec:concl}.

\subsection{Lagr.~\pref{eq:lagrD3f}} This case comprises a big family of systems depending on two arbitrary functions~$a(q)$, $b(q)$.  Surprisingly, this family satisfies many integrability conditions, and with our present computing resources (and algorithms) we cannot restrict the class further. The generic case, with~$a'(q)\neq0$,~\emph{does not admit higher symmetries} of differential order up to five, so it seems nonintegrable. Most probably, the obstruction to integrability~$c_i$ that restricts the family is of very high level~$i$~\cite{NovPC}. In order to illustrate this case we suppose that~$a(q)=\alpha$ is constant and show how in this case the system is indeed integrable.

With~$a(q)=\text{constant}=\alpha$, system~\pref{eq:lagrD3f} posseses the conserved densities~$q_{t}q_{x}+\alpha q_{x}^{2}$ and~$\frac{q_{xx}^{2}}{2q_{x}^{4}}-\frac{q_{t}^{2}}{2}+\frac{\alpha^{2}q_{x}^{2}}{2}+d(q)$ with associated operators
\[
\mathcal{C}_{1}=\left(2\alpha q_{x}+q_{t}\right)D+q_{x}D_{t},\quad\mathcal{C}_{2}=\frac{q_{xx}}{q_{x}^{4}}D^{2}+\left(\alpha^{2}q_{x}-2\frac{q_{xx}^{2}}{q_{x}^{5}}\right)D+d'(q)-q_{t}D_{t}
\]
and a recursion operator
\begin{multline*}
\R_1=\mathcal{L}_{2}=\frac{1}{q_{x}^{2}}D^{2}-\frac{5}{2}\frac{q_{xx}}{q_{x}^{3}}D-\frac{q_{xxx}}{q_{x}^{3}}+\frac{9}{4}\frac{q_{xx}^{2}}{q_{x}^{4}}+\frac{3}{4}q_{t}^{2}+\frac{3}{2}\alpha q_{x}q_{t}+\frac{3}{4}\alpha^{2}q_{x}^{2}-\frac{1}{2}d(q)
\\
-\left(\alpha q_{x}+\frac12q_{t}\right)D^{-1}\mathcal{C}{}_{1}+\frac{1}{2}q_{x}D^{-1}\mathcal{C}_{2}.
\end{multline*}
For a generic~$d(q)$ this equation has only higher symmetries of type~$s^{\rm b}_j$, $j\geq1$ generated by~$\R_1$ with a seed symmetry~$q_t$, i.e.~$\R_1(s^b_j)=s^b_{j+1}$, $j\geq0$. There are interesting subcases admitting more symmetries, depending on the form of~$d(q)$.

\paragraph{$\bullet$ The subcase with~$d''(q)=3d^2(q)/2$} This system additionally admits all generalised symmetries of the form~$s^a_i$, and more conservation laws. The recursion operator~$\R_1=\L_2$ is able to generate a standard symmetry structure starting with seed symmetry~$s^b_0=q_t$, and a second seed symmetry
\[s^a_1=\frac{q_{xxx}}{q_{x}^3}-\frac{9 q_{xx}^2}{4 q_{x}^4}-\frac{3 q_{t}^2}{4}-\frac{3}{2} \alpha  q_{x}q_{t}-\frac{3}{4} \alpha ^2 q_{x}^2+\frac43 d(q)
\]
with~$\R_1(s^a_i)=s^a_{i+1}$, $i\geq1$. There exists an additional independent recursion operator~$\R_2=\M_1$ that can be written in finite form as
\begin{multline*}
\R_2=\mathcal{M}_{1}=-\frac{1}{2}\frac{3q_{t}+\alpha q_{x}}{q_{x}^{2}}D^{2}-\left(\frac{3}{2}\frac{q_{xt}}{q_{x}^{2}}-\frac{5}{4}\frac{3q_{t}+\alpha q_{x}}{q_{x}^{3}}q_{xx}\right)D-s_{1}^{b}
\\
{}-\frac{1}{8}\left(2q_{t}+5\alpha q_{x}\right)d(q)+\left[\frac{1}{q_{x}}D-\frac{q_{xx}}{q_{x}^{2}}\right]D_{t}-\frac{1}{2}s_{1}^{a}D^{-1}\mathcal{C}_{1}+\frac{1}{2}q_{x}D^{-1}\mathcal{C}_{3}
\end{multline*}
where
\begin{multline*}
\mathcal{C}_3=\left(\frac{q_{xt}}{q_{x}^{3}}-\frac{3q_{t}-\alpha q_{x}}{2q_{x}^{4}}q_{xx}\right)\,D^{2}
\\
{}+\left(-\frac{3q_{xx}q_{xt}}{q_{x}^{4}}+3\frac{4q_{t}-\alpha q_{x}}{4q_{x}^{5}}q_{xx}^{2}+\frac{3\alpha q_{t}^{2}}{4}+\frac{3\alpha^{2}}{2}q_{x}q_{t}+\frac{3\alpha^{3}}{4}q_{x}^{2}\right)D-\frac{3}{4}q_{t}d(q)_{q}
\\
{}	+\left[\frac{q_{xx}}{q_{x}^{3}}D-\frac{3q_{xx}^{2}}{4q_{x}^{4}}+\frac{3q_{t}^{2}}{4}+\frac{3}{2}\alpha q_{t}q_{x}+\frac{3}{4}\alpha^{2}q_{x}^{2}-\frac{3d(q)}{4}\right]D_{t}
\end{multline*}
is the linearisation of the additional conserved density
\[\rho_3=\frac{q_{xx} q_{xt}}{q_{x}^3}-\frac{3 q_{t}-\alpha  q_x}{4 q_{x}^4}q_{xx}^2+\frac{q_{t}^3}{4}+\frac{3}{4} \alpha  q_{t}^2 q_{x}+\frac{3}{4} \alpha ^2 q_{t} q_{x}^2-\frac{3}{4} d(q) q_{t}+\frac{1}{4} \alpha ^3 q_{x}^3.
\]
The recursion operator~$\R_2$ acts on symmetries as~$\R_2(s^a_i)=s^b_{i+1}$ and~$\R_2(s^b_j)=s^a_{j+2}$, $i\geq1$, $j\geq0$ and~\pref{lem:one} implies that~$\M_1^2=\L_2^3$.

\paragraph{$\bullet$ The subcase with~$d'(q)=0$.} This is a subcase of~\pref{eq:lagrD5f}, and it admits a recursion operator $\mathcal{M}_{-1}$
\[\R=\M_{-1}=-\frac{\alpha}{2}q_{x}-q_t+D^{-1}\cdot\left(\alpha q_{x}+\frac{1}{2}q_t\right) D
+\left[\frac{1}{2}q_{x}D^{-1}+D^{-1}\cdot\frac{1}{2}q_{x}\right]D_{t}.
\]
that is a square root of~$\mathcal{L}_{2}$ and generates a standard symmetry structure from~$q_t$.

\subsection{Lagr.~\pref{eq:lagrD5f}} This system admits the conserved densities~$-\epsilon q_{x}^{2}-q_{x}q_t$, $\frac{\gamma^{2}}{q_{x}}+\gamma q_t$
with associated operators
\[\mathcal{C}_{1}=\left(-2	\epsilon q_x-q_t\right)D - q_xD_t,\qquad \mathcal{C}_{2}=-\frac{\gamma^2}{q_x^2}D+\gamma D_t.
\]
A recursion operator is~$\mathcal{M}_{-1}$
\[\R=\M_{-1}=-\frac{\epsilon}{2}q_{x}-q_t-\frac12D^{-1}\mathcal{C}_{1}+\frac1{2\gamma}q_xD^{-1}\mathcal{C}_{2}.
\]
Applied to~$q_t$ this recursion operator generates the standard family of symmetries.

\subsection{Lagr.~\pref{eq:lagrD6f}} There is a recursion operator~$\mathcal{M}_{-1}$:
\[\R=\mathcal{M}_{-1}=\epsilon q_{x}-2\epsilon D^{-1}\cdot q_{xx}-D^{-1}\cdot q_{xt}
+D^{-1}\cdot q_{x}D_{t}.
\]
Applied to~$q_t$ this recursion operator generates a standard family of symmetries.

\subsection{Lagr.~\pref{eq:lagrD7f}.} The equation with generic functions~$a(q)$, $b(q)$, $d_2(q)$ and~$d_3(q)$ admits a, notably, \emph{differential} recursion operator
\begin{equation}\label{eq:rom1}
\R_1=\mathcal{M}_{0}=-D\cdot\frac{q_t}{q_{x}}+D_{t}
\end{equation}
Applied to point symmetry~$q_x$ the recursion operator yields~$0$, i.e.~$\M_0(q_x)=0$. Applied to point symmetry~$q_t$ it yields a higher symmetry
\begin{multline*}
s_2^4\coloneqq \frac{q_{xxxx}}{a^{4}q_{x}^{4}}-\frac{10q_{xx}q_{xxx}}{a^{4}q_{x}^{5}}-\frac{8a'q_{xxx}}{a^{5}q_{x}^{3}}+\frac{2bq_{xt}}{q_{x}}+\frac{15q_{xx}^{3}}{a^{4}q_{x}^{6}}+\frac{24a'q_{xx}^{2}}{a^{5}q_{x}^{4}}
\\
{}-\frac{2bq_{t}q_{xx}}{q_{x}^{2}}+b'q_{t}-\frac{2\left(2aa''+a^{6}d_{2}-10(a')^{2}\right)q_{xx}}{a^{6}q_{x}^{2}}+2d'_{2}+d_{3}'q_{x}.
\end{multline*}
This symmetry does not belong to one of the categorised symmetries, allowing the definition of a third class of symmetries
\[\text{symmetries of type c:}\quad s^{\textrm{c}}_{k}=e_kq_{2k}+\cdots,\quad k=0,1,\ldots.
\]
The recursion operator~$\M_0$ generates an infinite hierarchy of higher symmetries of the type~$s^b_{2j}$, $s^c_{2k}$ with~$j\geq1$ and~$k\geq2$.

This family of systems has subcases admitting more symmetries and recursion operators. We do not give here a full classification of subcases, but we show as an example the following system with a maximum number of independent symmetries:
\begin{equation}
\L=\frac{q_t^2}{2q_x}+\beta\frac{ q}{q_x}\,q_t+\frac{q^4}{2q_x^5}q_{xx}^2
+\delta\frac{q^2}{q_x}.
\end{equation}
Some of its symmetries are ($\Delta\coloneqq\sqrt{8\delta+9}$)
\begin{gather*}
q,\quad q_{x},\quad q_t,\quad\frac{q_{x}}{q},\quad q^{-(1+\Delta)/2}q_{x},\quad q^{-(1-\Delta)/2}q_{x},\\
\frac{q^{2}}{q_{x}^{2}}q_{xx},\quad\frac{q}{q_{x}}q_{xt}-\frac{qq_t}{q_{x}^{2}}q_{xx}
\\
\frac{q^{3}q_{xx}q_{xt}}{2q_{x}^{4}}-\frac{q^{3}q_tq_{xx}^{2}}{2q_{x}^{5}}-\frac{3q^{2}q_tq_{xx}}{4q_{x}^{3}}-\frac{\beta q_t^{2}}{2q_{x}}+\frac{3qq_t}{2q_{x}}+\frac{\beta\Delta^{2}q^{2}}{8q_{x}},
\end{gather*}
All higher symmetries can be generated with~$\M_0$ as in~\pref{eq:rom1} and~$\L_1$ i.e.
\[\R_1=\mathcal{M}_{0}=-D\cdot\frac{q_t}{q_{x}}+D_{t},\qquad\R_2=\mathcal{L}_{1}=D\cdot\frac{q}{q_{x}}.
\]

\subsection{Lagr.~\pref{eq:lagrD10f}} This system is equivalent to the Landau-Lifshitz model~\pref{eq:eqLL}, admits conserved densities
\[
\rho_{2}=\frac{q_{x}q_t}{q_{x}^{2}-1},\qquad
\rho_{3}=\frac{q_{xx}^{2}-q_t^{2}}{2(q_{x}^{2}-1)}+d(q)q_{x}^{2}-\frac{1}{3}d(q)
\]
and~$\mathcal{L}_2$ is a recursion operator
\[\R_1=\mathcal{L}_{2}=D^2-\frac{q_xq_{xx}}{q_x^{2}-1}D+\frac23d(q)+q_tD^{-1}\cdot\mathcal{C}_2-q_xD^{-1}\cdot\mathcal{C}_3
\]
where
\[\mathcal{C}_2=-\frac{q_x^2+1}{(q_{x}^{2}-1)^2}q_tD+\frac{q_x}{q_{x}^{2}-1}D_t
\]
and
\[\mathcal{C}_3=\frac{q_{xx}}{q_x^2-1}D^2+\left[\frac{q_t^2-q_{xx}^2}{(q_x^2-1)^2}q_x+2d(q)q_x\right]D+\frac13\left(3q_x^2-1\right)d'(q)-\frac{q_t}{q_x^2-1}D_t.
\]
The recursion operator~$\L_2$ generates a standard symmetry hierarchy starting from seed symmetries~$q_x$ and~$q_t$.

The subcase with~$d(q)=\delta=\text{constant}$ admits another first-order conserved density
\[\rho_1=\frac{q_t}{1-q_{x}^{2}}
\]
that allows~$\M_{-1}$ to be written in finite form as
\[\R_2=\M_{-1}=D^{-1}\mathcal{C}_{1}-q_xD^{-1}\mathcal{C}_{2}
\]
with
\[\mathcal{C}_{1}=\frac{2q_xq_t}{(1-q_{x}^{2})^2}D+\frac{1}{1-q_{x}^{2}} D_t.
\]
In this subcase, the symmetry hierarchy can be generated by repeatedly applying~$\M_{-1}$ to~$q_x$. \pref{lem:one} implies that~$\R_1=\R_2^2$.

\section{Final conclusions and discussion} \label{sec:concl}

A universal scheme for studying the symmetry algebras of non-evolutionary equations of the form~\pref{eq:eq} has been developed. The results given in section~4 for equations~\pref{eq:eq4} are easily extensible to equations~\pref{eq:eq}, as follows. An integrable equation~\pref{eq:eq} admits two types of formal recursion operators,~$\L_i$ and~$\M_j$, $i,j\in\Z$. Simple homogeneity arguments proved~\pref{thm:one}, which specifies that for integrable equations~\pref{eq:eq} with~$m=4$, $n=1$, the algebra of formal recursion operators is generated by two operators~$\L_1$, $\M_{-2}$. The operator~$\M_{-2}$ is a square root of the identity that relates operators~$\L_i$ with operators~$\M_j$, duplicating the structure found in~\cite{HSS}. In that paper it was stated that the algebra of equations with~$m=3$, $n=1$ is generated by only one operator,~$\M_{-1}$. The homogeneity weights~\pref{eq:wlm} imply that in this case
\[
\weight{\mathcal{L}_{i}}=\frac{2i}{3}w_{t},\quad\weight{\mathcal{M}_{j}}=\frac{2j+3}{3}w_{t}
\]
so~$\weight{\M_{-1}}=w_t/3$,~$\M_{-1}^{2k}=\L_1^k$, $\M_0=\M_{-1}^3$ and~$\M_{-1}$ is able to generate all powers~$\M_{0}^j\L_1^i$. For a general equation~\pref{eq:eq}
\[\weight{\mathcal{L}_{i}}=\frac{2i}{m}w_{t},\quad\weight{\mathcal{M}_{j}}=\frac{2j+m}{m}w_{t}.
\]
Thus, if~$m=2k+1$ is odd, the operator~$\M_{-k}$ generates all the algebra. If~$m=2k$ is even, the algebra is generated by two independent operators, e.g.~$\L_1$ and~$\M_{-k}$, being~$\M_{-k}$ a square root of the identity operator.
The scheme based on homogeneity considerations developed here  seems to be applicable to even more general equations of the type~$u_{t\cdots t}=f$.

The classification of Lagrangian systems as been given as an application of the theory. In families~\pref{eq:lagrD3f} and~\pref{eq:lagrD7f} integrability is only checked for special values of some arbitrary functions on the Lagrangians. We have the impression that there should be very high order obstructions to integrability refining the form of such arbitrary functions. We were unable to reach explicitly these higher orders because of their very demanding need of computational power. The significant subcases we have provided are found by requiring the existence of higher symmetries of order up to 5, and they probably exhaust the integrable cases. It would also be important to study the novelty of some of the integrable equations generated in this paper. This necessarily requires a study of possible differential substitutions or parameterisations of the equations that could relate them to simpler equations~(cf.~\cite{HSS}). But this theory is not fully developed and it merits a separate work.

Our algorithms were based on integration of differential functions, i.e.~detecting when a given differential expression is a total derivative. This allowed us to study non-polynomial equations that are usually out of the direct reach of very powerful symbolic techniques~\cite{MNW,NovWang}, directly applicable only to homogeneous polynomial systems. But our universal scheme, based on using differential polynomials in the symbols~$u_i$, $v_j$, $l_i$ and~$m_j$, produces expressions that are quasipolynomial and homogeneous. It seems feasible to develop a symbolic formalism mixing both approaches allowing to work only at an algebraic level, without derivations or integrations.

Finally, we note that the use of other representations like those in~\cite{MNW,NovWang} allow for the introduction of some nonlocalities that appear in some important integrable systems like the Boussinesq equation. In this work we have only concentrated on equations admitting local recursion operators~\pref{eq:scfro}. A further development would be to extend the theory to include nonlocalities.

\section*{Acknowledgements}
RHH would like to thank A.B.~Shabat and V.V. Sokolov for useful discussions and for suggesting the problem treated in this paper. The authors acknowledge partial support from ``Ministerio de Econom\'{i}a y Competitividad'' (MINECO, Spain) under grant MTM2016-79639-P (AEI/FEDER, EU). 

\appendix

\section{Integrability conditions}
\label{app:tableint}

We list here the first conservation laws~$D_{t}\rho_{i}=D_{x}\sigma_{i}$, $i=0,..,5$ that are integrability conditions for an equation of the form~\prettyref{eq:eq4}
\begin{gather*}
\rho_{0}=\frac{1}{u_{4}^{1/4}},\qquad\rho_{1}=\frac{u_{3}}{u_{4}},\qquad\rho_{2}=\frac{v_{1}}{\sqrt{u_{4}}}-2\frac{\sigma_{0}}{u_{4}^{1/4}}
\\
\rho_{3}=-\frac{3v_{1}D_{x}u_{4}}{2u_{4}^{5/4}}-\frac{2\sigma_{1}}{u_{4}^{1/4}}-\frac{4v_{0}}{u_{4}^{1/4}}+\frac{u_{3}v_{1}}{u_{4}^{5/4}}
\\
\rho_{4}=\frac{5\left(D_{x}u_{4}\right)^{2}}{u_{4}^{7/4}}-\frac{12u_{3}D_{x}u_{4}}{u_{4}^{7/4}}-\frac{16\sigma_{0}^{2}}{\sqrt[4]{u_{4}}}-\frac{16\sigma_{2}}{\sqrt[4]{u_{4}}}-\frac{4v_{1}^{2}}{u_{4}^{3/4}}-\frac{32u_{2}}{u_{4}^{3/4}}+\frac{12u_{3}^{2}}{u_{4}^{7/4}}
\\
\rho_{5}=\frac{5\sigma_{0}\left(D_{x}u_{4}\right)^{2}}{u_{4}^{7/4}}-\frac{12\sigma_{0}u_{3}\left(D_{x}u_{4}\right)}{u_{4}^{7/4}}-\frac{4\sigma_0v_{1}^{2}}{u_{4}^{3/4}}-\frac{16\sigma_0^{3}}{\sqrt[4]{u_{4}}}
\\
{}-\frac{16\sigma_2\sigma_0}{\sqrt[4]{u_{4}}}-\frac{32\sigma_0u_{2}}{u_{4}^{3/4}}+\frac{12\sigma_0u_{3}^{2}}{u_{4}^{7/4}}-\frac{\sigma_4}{\sqrt[4]{u_{4}}}
\end{gather*}

\end{document}